\documentclass{llncs}

\usepackage{makeidx}  

\usepackage[utf8]{inputenc}
\usepackage{times}
\usepackage{amsmath,amssymb}
\usepackage{tikz}
\usepackage{hyperref}
\usepackage[linesnumbered]{algorithm2e}
\usetikzlibrary{positioning}

\usepackage{graphicx}

\usepackage{cleveref}

\usepackage[normalem]{ulem}

\usepackage[
backend=biber,
sorting=nyt,
style=numeric,
uniquename=mininit,
uniquelist=minyear
]{biblatex}
\newcommand{\qedhere}{\qed}

\newcommand{\fr}{\mathbf{R}}
\newcommand{\frset}{\mathcal{R}} 

\newcommand{\formula}[1]{\mathbf{#1}}
\newcommand{\prop}{\formula{P}}
\newcommand{\init}{\formula{I}}
\newcommand{\trans}{\formula{T}}
\newcommand{\transpdrc}{\trans^{\textsc{PDRC}}}

\newcommand{\middot}{$\cdot$}

\begin{document}
\RestyleAlgo{boxruled}

\title{A Supervisory Control Algorithm Based on Property-Directed Reachability\thanks{The final publication is available at Springer via \url{https://doi.org/10.1007/978-3-319-70389-3_8}.}}

\authorrunning{Jonatan Kilhamn et al.} 
%
\author{Koen Claessen\inst{1} \and Jonatan Kilhamn\inst{1} \and Laura Kov{\'a}cs\inst{1}\inst{3} \and Bengt Lennartson\inst{2}}
\institute{Department of Computer Science and Engineering, \and Department of Electrical Engineering, \\ Chalmers University of Technology \\
\and Faculty of Informatics, 
Vienna University of Technology \\
\texttt{\{koen, jonkil, laura.kovacs, bengt.lennartson\}@chalmers.se}}

\date{\today}



\maketitle

\pagestyle{headings}  

\begin{abstract}    



We present an algorithm for synthesising a controller (supervisor) for a discrete event system (DES) based on the property-directed reachability (PDR) model checking algorithm. The discrete event systems framework is useful in both software, automation and manufacturing, as problems from those domains can be modelled as discrete supervisory control problems. As a formal framework, DES is also similar to domains for which the field of formal methods for computer science has developed techniques and tools. In this paper, we attempt to marry the two by adapting PDR to the problem of controller synthesis. The resulting algorithm takes as input a transition system with forbidden states and uncontrollable transitions, and synthesises a safe and minimally-restrictive controller, correct-by-design. We also present an implementation along with experimental results, showing that the algorithm has potential as a part of the solution to the greater effort of formal supervisory controller synthesis and verification.



\end{abstract}

\section{Introduction}

Supervisory control theory deals with the problems of finding and verifying controllers to given systems. One particular problem is that of controller synthesis: given a system and some desired properties---safety, liveness, controllability---automatically change the system so that it fulfills the properties. There are several approaches to this problem, including ones based on binary decision diagrams (BDD) \cite{BDD1,BDD2}, predicates \cite{predicate-transformers-sct} and the formal safety checker IC3 \cite{Shoaei2014}. 

In this work we revisit the application of IC3 to supervisory control theory. Namely, we present an algorithm for synthesising a controller (supervisor) for a discrete event system (DES), based on property-directed reachability~\cite{PDR} (PDR, a.k.a. the method underlying IC3~\cite{Bradley2011}). Given a system with a safety property and uncontrollable transitions, the synthesised controller is provably \emph{safe}, \emph{controllable} and \emph{minimally restrictive}~\cite{DESControl}.

\subsection{An illustrative example}

Let us explain our contributions by starting with an example. \Cref{fig:example} shows the transition system of a finite state machine extended with integer variables $x$ and $y$. The formulas on the edges denote guards (transition cannot happen unless formula is true) and updates (after transition, $x$ takes the value specified for $x'$ in the formula). This represents a simple but typical problem from the domain of control theory, and is taken from~\cite{Shoaei2015}.

\begin{figure}
\centering
\begin{tikzpicture}
\begin{scope}[every node/.style={thick}]
    \node[] (I) at (0,1) {$x=0,y=0$};
    \node[shape=circle,draw=black] (A) at (0,0) {$l_0$};
    \node[shape=circle,draw=black] (B) at (0,-2) {$l_1$};
    \node[shape=circle,draw=black] (C) at (3,0) {$l_2$};
    \node[shape=circle,draw=black] (D) at (3,-2) {$l_3$};
    \node[shape=circle,draw=black] (E) at (3,-4) {$l_4$};
    \node[shape=circle,draw=black,style=dashed] (F) at (7,-2) {$l_5$};
\end{scope}

\begin{scope}[->, 
              every edge/.style={draw=black,thick}]
    \path [->](I) edge (A);
    \path [->](A) edge node[left] {$b:y'=1$} (B);
    \path [->](A) edge node[above] {$a:y'=2$} (C);
    \path [->](B) edge[bend left=20] node[above] {$a:\top$} (D);
    \path [->](D) edge[bend left=20] node[below] {$c:x'=x+1$} (B);
    \path [->](C) edge node[right] {$b:\top$} (D);
    \path [->](D) edge node[right] {$\alpha:y=2\land x\le 2$} (E);
    \path [->](D) edge node[above] {$\alpha:y=2\land x>2$} (F);
    \path [->](E) edge[loop right] node[right] {$\omega:\top$} (E);
\end{scope}
\end{tikzpicture}
\caption{The transition system of the example.}
\label{fig:example}
\end{figure}
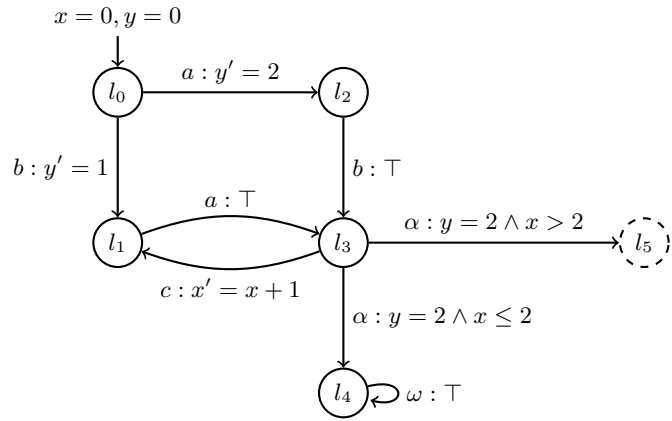

In a controller synthesis problem, a system such as this is the input. The end result is a restriction of the original system, i.e. one whose reachable state space is a subset of that of the original one. In this extended finite state machine (denoted as EFSM) representation, this is written as new and stronger guard formulas on some of the transitions.

Our example  has two more features: the location $l_5$, a dashed circle in the figure, is \emph{forbidden}, while the event $\alpha$ is \emph{uncontrollable}. The latter feature means that the synthesised controller must not restrict any transition marked with the event $\alpha$.

To solve this problem, we introduce an algorithm based on PDR~\cite{PDR} used in a software model checker (Section~\ref{sec:algo}). Intuitively, what our algorithm does is to incrementally build an inductive invariant which in turn implies the safety of the system. This invariant is constructed by ruling out paths leading into the bad state, either by proving these bad states unreachable from the initial states, or by making them unreachable via strengthening the guards.

In our example, the bad state $l_5$ is found to have a preimage under the transition relation $\trans$ in $l_3 \land y=2 \land x>2$. The transition from $l_3$ to $l_5$ is uncontrollable, so in order to guarantee safety, we must treat this prior state as unsafe too. The transitions leading into $l_3$ are augmented with new guards, so that the system may only visit $l_3$ if the variables make a subsequent transition to $l_5$ impossible. By applying our work, we refined \Cref{fig:example} with the necessary transition guards and a proof that the new system is safe. We show the refined system obtained by our  approach in \Cref{fig:example_controlled}. 

\begin{figure}
\centering
\begin{tikzpicture}
\begin{scope}[every node/.style={thick}]
    \node[] (I) at (0,1) {$x=0,y=0$};
    \node[shape=circle,draw=black] (A) at (0,0) {$l_0$};
    \node[shape=circle,draw=black] (B) at (0,-2) {$l_1$};
    \node[shape=circle,draw=black] (C) at (3,0) {$l_2$};
    \node[shape=circle,draw=black] (D) at (3,-2) {$l_3$};
    \node[shape=circle,draw=black] (E) at (3,-4) {$l_4$};
    \node[shape=circle,draw=black,style=dashed] (F) at (7,-2) {$l_5$};
\end{scope}
\begin{scope}[->, 
              every edge/.style={draw=black,thick}]
    \path [->](I) edge (A);
    \path [->](A) edge node[left] {$b:y'=1$} (B);
    \path [->](A) edge node[above] {$a:y'=2$} (C);
    \path [->](B) edge[bend left=20] node[above] {$a:\mathbf{y \ne 2 \lor x\ngtr 2}$} (D);
    \path [->](D) edge[bend left=20] node[below] {$c:x'=x+1$} (B);
    \path [->](C) edge node[right] {$b:\mathbf{y \ne 2 \lor x \ngtr 2}$} (D);
    \path [->](D) edge node[right] {$\alpha:y=2\land x\le 2$} (E);
    \path [->](D) edge node[above] {$\alpha:y=2\land x>2$} (F);
    \path [->](E) edge[loop right] node[right] {$\omega:\top$} (E);
\end{scope}
\end{tikzpicture}
\caption{The transition system from the example, with guards updated to reflect the controlled system.}
\label{fig:example_controlled}
\end{figure}
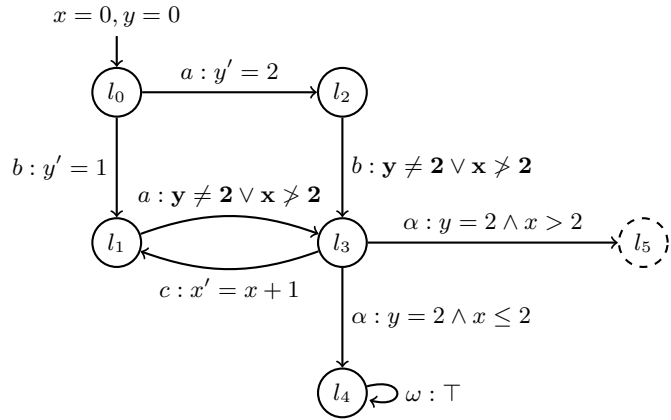

\subsection{Our Contributions}
\begin{enumerate}
\item
In this paper we present a novel algorithm based on PDR for controller synthesis (Section~\ref{sec:algo}) and prove correctness and termination of our approach (Section~\ref{sec:props}). To the best of our knowledge, PDR has not yet been applied to supervisory control systems in this fashion.  We prove that our algorithm terminates (given finite variable domains) and that the synthesised controller is safe, minimally-restrictive, and respects the controllability constraints of the system. Our algorithm encodes system variables in the SAT domain; we however believe that our work can be extended by using satisfiability modulo theory (SMT) reasoning instead of SAT.   

\item We implemented our algorithm in the model checker Tip~\cite{Tip}. 
We evaluated our implementation on a number of  control theory problems and give practical evidence of the benefits of our work (see Section~\ref{sec:exp}).
\end{enumerate}

\section{Background}\label{sec:prelim}


We use standard terminology and notation from first-order logic (FOL) and  restrict  formulas mainly to quantifier-free formulas. We reserve $\prop, \fr, \trans, \init $ to denote formulas describing, respectively, safety properties, ``frames'' approximating reachable sets, transition relations and initial properties of control systems; all other formulas will be denoted with $\phi, \psi$, possibly with indices.  We write variables as $x,y$ and sets of variables as $X,Y$. A \emph{literal} is an atom or its negation, a \emph{clause} a disjunction of literals, and a \emph{cube} a conjunction of literals. We use $\frset$ to denote a set of clauses, intended to be read as the conjunction of those clauses. When a formula ranges over variables in two or more variable sets, we take $\phi(X,Y)$ to mean $\phi(X \cup Y)$.

For every variable $x$ in the control system, we assume the existence of a unique variable $x'$ representing the \emph{next-state value} of $x$. Similarly, the set $X'$ is the set $\{x'|x \in X\}$. As we may sometimes drop the variable set from a formula if it is clear from the context, i.e. write $\phi$ instead of $\phi(X)$, we take $\phi'$ to mean $\phi(X')$ in a similar fashion.

\subsection{Modelling Discrete Event Systems}
\label{sec:des-models}

A given DES can be represented in several different ways. The simple, basic model is the finite state machine (FSM) \cite{state-machines}. A state machine is denoted by the tuple $G = \langle Q, \Sigma, \delta, Q^i\rangle$, where $Q$ is a finite set of states, $\Sigma$ the finite set of events (alphabet), $\delta \subseteq Q \times \Sigma \times Q$ the transition relation, and $Q^i \subseteq Q$ the set of initial states.

In this notation, a \emph{controller} can be represented as a function $C:Q \rightarrow 2^\Sigma$ denoting which events are enabled in a given state. For any $\sigma \in \Sigma$ and $q \in Q$, the statement $\sigma \in C(q)$ means that the controller allows transitions with the event $\sigma$ to happen when in $q$; conversely, $\sigma \notin C(q)$ means those transitions are prohibited.

\noindent{\bf Extended Finite State Machine.} 
The state machine representation is general and monolithic. In order to more intuitively describe real supervisory control problems, other formalisms are also used. Firstly, we have the extended finite state machine (EFSM), which is an FSM extended with FOL formulas over variables. In effect, we split the states into \emph{locations} and \emph{variables}, and represent the system by the tuple $A = \langle X, L, \Sigma, \Delta, l^i, \Theta \rangle$. Here, $X$ is a set of variables, $L$ a set of locations, $\Sigma$ the alphabet, $\Delta$ the set of transitions, $l^i \in L$ the initial location and $\Theta(X)$ a formula describing the initial values of the variables.

A transition in $\Delta$ is now a tuple $\langle l,a,m \rangle$ where $l,m$ are the entry and exit locations, respectively, while the action $a=(\sigma,\phi)$ consists of the event $\sigma \in \Sigma$ and $\phi(X,X')$. The interpretation of this is that the system can make the transition from $l$ to $m$ if the formula $\phi(X,X')$ holds. Since the formula can include next-state variables---$\phi$ may contain arbitrary linear expressions over both $X$ and $X'$---the transition can specify updated variable values for the new state.

We have now defined almost all of the notation used in the example in \Cref{fig:example}. In the figure, we write $\sigma:\phi$ to denote the action $(\sigma,\phi)$. Furthermore, the figure is simplified greatly by omitting next-state assignments on the form $x'=x$, i.e. $x$ keeping its current value. If a variable does not appear in primed form in a transition formula, that formula is implied to have such an assignment.

\noindent{\bf Symbolic Representation.} 
Moving from FSM to EFSM can be seen as ``splitting'' the state space into two spaces: the locations and the variables. A given feature of an FSM can be represented as either one (although we note that one purpose for using variables is to easier extend the model to cover an infinite state space). Using this insight we can move to the ``other extreme'' of the \emph{symbolic transition system} (STS): a representation with only variables and no locations.

The system is here represented by the tuple $S_A = \langle \hat{X}, \trans(\hat{X},\hat{X}'), \init(\hat{X}) \rangle$ where $\hat{X}$ is the set of variables extended by two new variables $x_L$ and $x_\Sigma$ with domains $L$ and $\Sigma$, respectively. With some abuse of notation, we use event and variable names to denote formulas over those variables, such as $l_n$ for the literal $x_L=l_n$ and $\lnot \sigma$ for the literal $x_\Sigma \neq \sigma$. The initial formula $\init$ and transition formula $\trans$ are constructed from the corresponding EFSM representation as $\init(\hat{X})=(x_L=l^i) \land \Theta(X)$ and $\trans(\hat{X},\hat{X}') = \bigvee_{\langle l,(\sigma,\phi),m \rangle \in \Delta}(l \land \sigma \land \phi(X,X') \land m')$.

In this paper, we will switch freely between the EFSM and STS representations of the same system, depending on which is the best fit for the situation. Additionally, we will at times refer to $\hat{X}$ as only $X$, as long as the meaning is clear from context. In either representation, we will use \emph{state} to refer to a single assignment of location and variables, and \emph{path} for a sequence of states $s_0,s_1,...,s_k$.

\subsection{Supervisory Control}
\label{sec:desiderata}

The general problem of supervisory control theory is this: to take a transition system, such as the ones we have described so far, and modify it so that it fulfils some property which the unmodified system does not. There are several terms in this informal description that require further explanation.

The properties that we are interested in are generally \emph{safety}, \emph{non-blocking}, and/or \emph{liveness}, which can be seen as a stronger form of non-blocking. Controlling for a safety property means that in the controlled system, there should be no sequence of events which enables transitions leading from an initial state to a forbidden state.

Non-blocking and liveness are defined relative to a set of marked state. The former means that at least one such state is reachable from every state which is reachable from the initial states. The latter, liveness, implies non-blocking, as it is the guarantee that the system not only \emph{can} reach but \emph{will} return to a marked state infinitely often. In this work we have reduced the scope of the problem by considering only safety.

Furthermore, we talk about the property of \emph{controllability}. This is the notion that some events in a DES are uncontrollable, which puts a restriction on any proposed controller: in order to be valid, the transitions involving uncontrollable events must not be restricted. Formally, in an (E)FSM it is enough to split the alphabet into the uncontrollable $\Sigma_u \subseteq \Sigma$ and the controllable $\Sigma_c = \Sigma \setminus \Sigma_u$. In an STS, this is expressed by the transition relation taking the form $\trans = \trans_u \lor \trans_c$, where $\trans_c$ and $\trans_u$ include literals $x_L = \sigma$ for, respectively, only controllable and only uncontrollable events $\sigma$.

Finally there is the question of what form this ``controlled system'' takes, since a controller function $C:Q\rightarrow2^\Sigma$ can be impractical. A common method is that of designating a separate state machine as the supervisor, and taking the controlled system to be the \emph{synchronous composition} of the original system and the supervisor~\cite{Synch}. In short, this means running them both in parallel, but only allowing a transition with a shared event $\sigma$ to occur simultaneously in both sub-systems.

However, the formidable theory of synchronised automata is not necessary for the present work. Instead, we take the view that the controlled system is the original system, either in the EFSM or STS formulation, with some additions. 


In the EFSM case, the controlled system has the exact same locations and transitions, but additional guards and updates may be added. In other words, the controlled system augments each controllable transition by replacing the original transition formula $\phi$ with the new formula $\phi^s = \phi \land \phi^\text{new}$. The uncontrollable transitions are left unchanged. In the STS case, the new transition function is $\trans^S = \trans_u \lor \trans^S_c$ where $\trans^S_c=\trans_c \land \trans^\text{new}_c$. This way, all uncontrollable transitions are guaranteed to be unmodified in the controlled system.

Finally, a controlled system, regardless of which properties the controller is set out to guarantee, is often desired to be \emph{minimally restrictive} (eqiv. maximally permissive). The restrictiveness of a controlled system is defined as follows: out of two controlled versions $S_1$ and $S_2$ of the same original system $S$, $S_1$ is more restrictive than $S_2$ if there is at least one state, reachable under the original transition function $\trans$, which is reachable under $\trans^{S_2}$ but unreachable under $\trans^{S_1}$. A controlled system is minimally restrictive if no other (viable) controlled system exists which is less restrictive.
The word ``viable'' in brackets shows that one can talk about the minimally restrictive \emph{safe} controller, the minimally restrictive \emph{non-blocking} controller and so on; for each combination of properties, the minimally restrictive controller for those properties is different.


\section{PDRC: Property-Driven Reachability-Based Control}\label{sec:algo}


Property-driven reachability (PDR)~\cite{PDR} is a name for the method underlying IC3~\cite{Bradley2011}, used to verify safety properties in transition systems. In this paper we present Property-Driven Reachability-based Control (PDRC), which extends PDR from verifying safety to synthesising a controller which makes the system safe. In order to explain PDRC, we first review the main ingredients of PDR.

PDR works by successively \emph{blocking} states that are shown to lead to unsafe states in certain number of steps. Blocking a state at step $k$ here means performing SAT-queries to show that the state is unreachable from the relevant \emph{frame} $\fr_k$. A frame $\fr_k$ is a predicate over-approximating the set of states reachable from the initial states $\init$ in $k$ steps.

When a state is blocked---i.e. shown to be unreachable---the relevant frame is updated by excluding that state from the reachable-set approximation. If a state cannot be blocked at $\fr_k$, the algorithm finds its preimage $s$ and proceeds to block $s$ at $\fr_{k-1}$. If a state that needs to be blocked intersects with the initial states, the safety property of the system has been proven false. Conversely, if two adjacent frames $\fr_i,\fr_{i+1}$ are identical after an iteration, we have reached a fixed-point and a proof of the property $P$ in one of them entails a proof of $P$ for the whole system.


With PDRC, we focus on the step where PDR has found a bad cube $s$ (representing unsafe states) in frame $\fr_k$, and proceeds to check whether it is reachable from the previous frame $\fr_{k-1}$. If it is not, this particular cube was a false alarm: it was in the over-approximation of $k$-reachable states, but after performing this check we can sharpen that approximation to exclude $s$. If $s$ was reachable, PDR proceeds to find its preimage $t$ which is in $\fr_{k-1}$. Note that $t$ is also a bad cube, since there is a path from $t$ to an unsafe state. However, in a supervisory control setting, there is no reason not to immediately control the system by restricting all controllable transitions from $t$ to $s$. This observation is the basis of our PDRC algorithm.

\subsection{Formal Description of PDRC}
\label{sec:alg_description}

As PDRC is very similar to PDR, this description and the pseudocode procedures draw heavily from~\cite{PDR}.


Our PDRC algorithm is given in \Cref{alg:main}.
As input, we take a transition system that can be represented by a transition function $\trans(X,X') = \trans_c \lor \trans_u$, i.e. one where each possible transition is either controllable or uncontrollable; and a safety property $\prop(X)$. The variables in $X$ are boolean, in order to allow the use of a SAT solver -- although see \Cref{sec:smt} describing an extension from SAT to SMT. 

Throughout the run of the algorithm, we keep a \emph{trace}: a series of frames $\fr_i, 0 \le i \le N$. Each $\fr_i(X)$ is a predicate that over-approximates the set of states reachable from $\init$ in $i$ steps or less. $\fr_0 = \init$, where $\init$ is a formula encoding the initial states.

Each frame $\fr_i, i > 0$ can be represented by a set of clauses $\frset_i=\{c_{ij}\}_j$, such that $\bigwedge_j c_{ij}(X) = \fr_i(X)$. An empty frame $\frset_j=\{\}$ is considered to encode $\top$, i.e. the most over-approximating set possible.


We maintain the following invariants:

\begin{enumerate}
    \item $\fr_i \rightarrow \fr_{i+1}$
    \item $\fr_i \rightarrow \prop$, except for $i=N$
    \item $\fr_{i+1}$ is an over-approximation of the image of $\fr_i$ under $\trans$
\end{enumerate}

\label{sec:trace_invariants}

Starting with $N=1$ and $\frset_1= \{\}$, we proceed to do the first iteration of the blocking and propagation steps, as shown in \Cref{alg:main}.

The ``blocking step'' consists of the while-loop (lines \ref{main:block_start}--\ref{main:block_end}) of \Cref{alg:main}, and coming out of that loop we know that $\fr_N \rightarrow \prop$. The propagation step follows (lines \ref{main:prop_start}--\ref{main:prop_end}), and here we consider for each clause in some frame of the trace whether it also holds in the next frame.

Afterwards, we check for a fix-point in $\fr_i$; i.e. two syntactically equal adjacent frames $\fr_i = \fr_{i+1}$. Unless such a pair is found, we increment $N$ by $1$ and repeat the procedure.

\begin{algorithm}[t]
\SetAlgoLined
\tcp{finding and blocking bad states}
  \While{SAT$[\fr_N \land \lnot \prop]$}{ \label{main:block_start}
    extract a bad state $m$ from the SAT model\;
    generalise $m$ to a cube $s$\;
    recursively block $s$ as per \texttt{block($s,N$)}\; \label{main:block_call}
    \tcp{at this point $\fr$ and/or $\trans$ have been updated to rule out $m$}
  } \label{main:block_end}
  \tcp{propagation of proven clauses}
  add new empty frame $\fr_{N+1}$\; \label{main:prop_start}
  \For{$k \in [1,N]$ and $c \in \fr_k$}{
    \If{$\fr_k \vDash c'$}{ 
      add $c$ to $\fr_{k+1}$\;
    }
  } \label{main:prop_end}
\caption{Blocking and propagation for one iteration of $N$.}
\label{alg:main}
\end{algorithm}


The most important step inside the \textbf{while} loop is the call to \texttt{block} (line \ref{main:block_call}). This routine is shown in \Cref{alg:block}. Here, we take care of the bad states in a straight-forward way. First, we consider its preimage under the controllable transition function $\trans_c$ (line \ref{block:contr_preimage}). The preimage cube $t$ can be found by taking a model of the satisfiable query $\fr_{k-1} \land \lnot s \land \trans_c \land s'$ and dropping the primed variables. Each such cube encodes states from which a bad state is reachable in one step. Thus, we update the supervisor to disallow transitions from those bad states (line \ref{block:update}). This accounts for the first while-loop in \Cref{alg:block}.

The second while-loop (lines \ref{block:uncontr_start}--\ref{block:uncontr_end})is very similar, but considers the uncontrollable transitions, encoded by $\trans_u$, instead. If a preimage cube is found here, we cannot rule it out by updating the supervisor. That preimage instead becomes a bad state on its own, to be controlled in the previous frame $k-1$.

\begin{algorithm}[t]
\SetAlgoLined
\KwData{A cube $s$ and a frame index $k$}
\tcp{first consider the controllable transitions:}
  \While{SAT$[\fr_{k-1} \land 
               \lnot s \land
               \trans_c \land
               s']$}{
    extract and generalise a bad cube $t$ in the preimage of $\trans_c$\; \label{block:contr_preimage}
    update $\trans_c := \trans_c \land \lnot t$\; \label{block:update}
  }
\tcp{then consider the uncontrollable transitions:}
  \While{SAT$[\fr_{k-1} \land 
               \lnot s \land
               \trans_u \land
               s']$}{ \label{block:uncontr_start}
    \If{k=1}{ 
        throw \textbf{error: system uncontrollable}\;
    }
    extract and generalise a bad cube $t$ in the preimage of $\trans_u$\;
    call \texttt{block($t,k-1$)}\;
  } \label{block:uncontr_end}
add $\lnot s$ to $\fr_i, i \le k$\;
\caption{The blocking routine, which updates the supervisor.}
\label{alg:block}
\end{algorithm}

\begin{example}\noindent{\bf Example, revisited.} Recall the example in \Cref{fig:example}. Since it uses integer variables it seems to require an SMT-based version of PDRC. This particular example is so simple, however, that ``bit-blasting'' the problem into SAT by treating the proposition $x<i$ as a separate boolean variable for each value of $i$ in the domain of $x$ will yield the same solution.

PDRC requires $3$ iterations to completely supervise the system. In the first, the clause $\lnot l_5$ is added to the first frame $\frset_1$, after proving that it is not in the initial states. In the second, $\lnot l_5$ is found again but this time the uncontrollable transition from $l_3$ is followed backwards, and the clause $\lnot \alpha \lor \lnot l_3 \lor y \neq 2 \lor x \ngtr 2$ is also added to $\frset_1$, which allows us to add $\lnot l_5$ to $\frset_2$. Finally, in the third, the trace of preimages lead to the controllable transitions $l_1\rightarrow l_3$ and $l_2\rightarrow l_3$, and we add new guards to both (technically, we add new constraints to the transition function).

The updated system is the one shown in \Cref{fig:example_controlled}. The third iteration also proves the system safe, as we have $\fr_1 = \fr_2$. These frames then hold the invariant, $(\lnot \alpha \lor \lnot l_3 \lor y \neq 2 \lor x \ngtr 2) \land \lnot l_5$, which implies $\prop$ and is inductive under the updated $\trans$.

\end{example}

\subsection{Extension to SMT}
\label{sec:smt}

Our PDRC algorithm in \Cref{alg:main} uses SAT queries, and is straightforward to use with a regular SAT solver on systems with a propositional transition function. However, like in~\cite{Cimatti2012, GPDR} it is possible to extend it to other theories, such as Linear Integer Arithmetic, using an SMT solver. The SAT query in \Cref{alg:main} provides no diffuculty, but some extra thought is required for the ones in the blocking procedure, which follow this pattern:
\newline

\-\hspace{1cm} \textbf{while} \textit{SAT}$[\fr_i \land \lnot s \land \trans \land s']$ \textbf{do}

\-\hspace{1.5cm} extract and generalise a bad cube $t$ in the preimage of $\trans$;
\ \newline

If one only replaces the SAT solver by an SMT solver capable of handling the theory in question, one can extract a satisfying assignment of theory literals. However, each of these might contain both primed and unprimed variables, such as the next-state assignment $x' = x+1$.


These lines effectively ask the solver to generalise a state $m$---an assignment of theory literals satisfying some formula $\formula{F}$---into a more general cube $t$, ideally choosing the $t$ that covers the maximal amount of discrete states, while still guaranteeing $t \rightarrow \formula{F}$. In the SAT case, this is achieved by dropping literals of $t$ that do not affect the validity of $\formula{F}(t)$. An alternate method based on ternary simulation, that is useful when the query is for a preimage of a transition function $\trans$, is given in~\cite{PDR}. For the SMT case, however, the extent of generalisation depends on the theory and the solver.

In the worst case of a solver that cannot generalise at all, the algorithm is consigned to blocking a single state $m$ in each iteration. This means that the state space simplification gained from using a symbolic transition function in the first place is lost, since the reachability analysis checks states one by one. In conclusion, PDRC could be implemented for systems with boolean variables using a SAT-solver with no further issues, while an SMT version would require carefully selecting the right solver for the domain.
We leave this problem as an interesting task for future work. 







\section{Properties of PDRC}\label{sec:props}

In this section we prove the soundness and termination of our PDRC algorithm. 

\subsection{Termination}

\begin{theorem}
For systems with state variables whose domains are finite, the PDRC algorithm always terminates.
\end{theorem}

The termination of regular PDR is proven in~\cite{PDR}. In the case of an unsafe system---which for us corresponds to an uncontrollable system---the counterexample proving this must be finite in length, and thus found in finite time. In the case of a safe system, the proof is based on the following observations: that each proof-obligation (call to \texttt{block}) must block at least one state in at least one frame; that there are a finite number of frames for each iteration (value of $N$); that there are a finite number of states of the system; and that each $\fr_{i+1}$ must either block at least one more state than $\fr_i$, or they are equal.

All these observations remain true for PDRC, substituting ``uncontrollable'' for ``unsafe''. This means that the proof of termination from~\cite{PDR} can be used for PDRC with minimal modification.

\subsection{Correctness}

We claim that the algorithm described above synthesises a minimally restrictive safe controller for the original system.

\begin{theorem}
If there exists any safe controller for the system, the controller synthesised by the PDRC algorithm is safe.
\label{th:safety}
\end{theorem}

\begin{proof}
We prove \Cref{th:safety} by contradiction. Assume there is an unsafe state $s$, i.e. we have $\lnot \prop(s)$, that is reachable from an $\init$-state in $k$ steps. We must then have $k \ge N$, since invariant $(2)$ states that $\fr_i\rightarrow \prop,i<N$.
Let $M$ be the index of the discovered fix point $\fr_{M} = \fr_{M+1}$.

Invariant $(1)$ (from \Cref{sec:trace_invariants}) states that $\fr_i \rightarrow \fr_{i+1}$, and this applies for all values $0 \le i \le M$. Repeated application of this means that any state in any $\fr_i, i < M$ is also contained in $\fr_M$.

Invariant $(3)$ states that $\fr_{i+1}$ is an over-approximation of the image of $\fr_i$. This means that any state reachable from $\fr_M$ should be in $\fr_{M+1}$. Since $\fr_M = \fr_{M+1}$, such a state is also in $\fr_M$ itself. Repeated application of this allows us to extend the trace all the way to $\fr_k = \fr_{k-1} = \dots = \fr_M$.

Now, for the bad state $s$, regardless of the number of steps $k$ needed to reach it, we know that $s$ is contained in $\fr_k$ and therefore in $\fr_M$. Yet when the algorithm terminated it had at one point found $\fr_M \land \lnot \prop$ to be UNSAT. The state $s$, which is both in $\fr_M$ and $\lnot \prop$, would constitute a satisfying assignment to this query. This contradiction proves that $s$ cannot exist.
\qedhere
\end{proof}

\begin{theorem}
A controller synthesised by the PDRC algorithm is minimally restrictive.
\label{th:min_restrictive}
\end{theorem}

\begin{proof}
We prove \Cref{th:min_restrictive} also by contradiction. Assume there is a safe path $\pi = s_0,s1,\dots,s_k$ through the original system (with transition function $\trans$), which is not possible using the controlled transition function $\transpdrc$; yet there exists another safe, controllable supervisor represented by $\trans^S$ where $\pi$ is possible. By deriving a contradiction, we will prove that no such $\trans^S$ can exist.

Consider the first step of $\pi$ that is not allowed by $\transpdrc$; in other words, a pair $(s_i, s_{i+1})$ where we have $\lnot \transpdrc(s_i,s_{i+1})$ while we do have both $\trans^{S}(s_i,s_{i+1})$ and $\trans(s_i,s_{i+1})$. The only way that $\transpdrc$ is more restrictive than $\trans$ is due to strengthenings on the form $\transpdrc_c = \trans_c \land \lnot m$, for some cube $m$. This means that $s_i$ must be in some cube $m$ that PDRC supervised in this fashion.

This happened inside a call \texttt{block($m,j$)}. Since $\pi$ is safe, this call cannot have been made because $m$ itself encoded unsafe states. Instead, there must have been a previous call \texttt{block($n,j+1$)}, where $m$ is a minterm of the preimage of $n$ under $\trans_u$. This cube $n$ is either itself a bad cube, or it can be traced to a bad cube by following the trace of \texttt{block} calls. Since each step in this \texttt{block} chain only uses $\trans_u$, we can find a series of uncontrollable transitions, starting in some $\tilde{s}_{i+1} \in n$, leading to some cube $p$ which is a generalisation of a satisfying assignment to the query $\fr_N \land \lnot \prop$.

This proves that $\trans^S$, whose $\trans^S_c$ does not restrict transitions from $s_i$, allows for the system to enter a state $\tilde{s}_{i+1}$, from which there is an uncontrollable path to an unsafe state. This contradicts the assumption that $\trans^S$ was safe, proving that the combination of $\pi$ and $\trans^S$ cannot exist. This proves that the controller encoded by $\transpdrc$ is minimally restrictive.
\qedhere
\end{proof}



\section{Implementation}

We have implemented a prototype of PDRC in the model checker Tip (Temporal Inductive Prover~\cite{Tip}). The input format supported by Tip is AIGER~\cite{AIGER}, where the transition system is represented as a circuit, which is not a very intuitive way to view an EFSM or STS. For this reason, our prototype also includes Haskell modules for creating a transition system in a control-theory-friendly representation, converting it to AIGER, and using the output from the Tip-PDRC to reflect the new, controlled system synthesised by PDRC. Finally, it also includes a parser from the .wmod format used by WATERS and Supremica~\cite{WATERS}, into our Haskell representation. 
Altogether, our implementation consists of about $150$ lines of code added or changed in the Tip source, and about $1600$ lines of Haskell code. Our tools, together with the benchmarks we used, is available through \href{https://github.com/JonatanKilhamn/supermini}{github.com/JonatanKilhamn/supermini} and \href{https://github.com/JonatanKilhamn/tipcheck}{github.com/JonatanKilhamn/tipcheck}.

When converting transition systems into circuits, certain choices have to be made. Our encoding allows for synchronised automata with \emph{one-hot-encoded} locations (e.g. location $l_3$ out of $5$ is represented by the bits $[0,0,1,0,0]$) and \emph{unary-encoded} integer variables (e.g. a variable ranging from $0$ to $5$ currently having the value $3$ is represented by $[1,1,1,0,0]$). Each of these encoding has a corresponding invariant: with one-hot, exactly one bit must be set to $1$; with unary, each bit implies the previous one. However, these invariants need not be explicitly enforced by the transition relation (i.e. as guards on every transition), rather, it is enough that they are preserved by all variable updates.

It should be noted that although the PDRC on a theoretical level works equally well on STS as EFSM, our implementation does assume the EFSM division between locations and variables for the input system. However, our implementation retains the generality of PDRC in how the state space is explored---the algorithm described in \Cref{sec:algo} is run on the circuit representation, where the only difference between the location variable $x_L$ and any other variable is the choice of encoding.

\section{Experiments}
\label{sec:exp}

For an empirical evaluation, we ran PDRC on several standard benchmark problems: the extended dining philosophers (EDP)~\cite{Supremica2008}, the cat and mouse tower (CMT)~\cite{Supremica2008} and the parallell manufacturing example (PME)~\cite{PME}. The runtimes of these experiments are shown in \Cref{tab:results} below. The benchmarks were performed on a computer with a $2.7$ GHz Intel Core i5 processor and 8GB of available memory.

\subsection{Problems}

For the dining philosophers, EDP$(n,k)$ denotes the problem of synthesising a safe controller for $n$ philosophers and $k$ intermediary states that each philosopher must go through between taking their left fork and taking their right one. The transition system is written so that all philosophers respect when their neighbours are holding the forks, except for the even-numbered ones who will try to take the fork to their left even if it is held, which leads (uncontrollably) to a forbidden state.

For the cat and mouse problem, CMT$(n,k)$ similarly denotes the problem with $n$ floors of the tower, $k$ cats and $k$ mice. Again, the transition system already prohibits cats and mice from entering the same room (forbidden state) except by a few specified uncontrollable pathways.

Finally, the parallel manufacturing example (PME) represents an automated factory, with an industrial robot and several shared resources. It differs from the other in that its scale comes mainly from the number of different synchronised automata. In return, it does not have a natural parameter that can be set to higher values to increase the complexity further.

\subsection{Results}

We compare PDRC to Symbolic Supervisory Control using BDD (SC-BDD)~\cite{BDD1,BDD2}, which is implemented within Supremica. We wanted to include the Incremental, Inductive Supervisory Control (IISC) algorithm~\cite{Shoaei2014}, which also uses PDR but in another way. However, the IISC implementation from~\cite{Shoaei2014} is no longer maintained. Despite this failed replication, we include figures for IISC taken directly from \cite{Shoaei2014}---with all the caveats that apply when comparing runtimes obtained from different machines. 
\Cref{tab:results} shows runtimes, where the problems are denoted as above and ``$\times$'' indicates time-out (5 min). The parameters for EDP and CMT were chosen to show a wide range from small to large problems, while still mostly choosing values for which \cite{Shoaei2014} reports runtimes for IISC. We see that while SC-BDD might have the advantage on certain small problems, PDRC quickly outpaces it as the problems grow larger.

\begin{table}
\centering
\caption{Performance of PDRC (our contribution), SC-BDD and IISC on standard benchmark problems. Note that the IISC implementation was not reproducible by us; the numbers here are lifted from \cite{Shoaei2014}. ``$\times$'' indicates timeout ($5$ min), and ``--'' means this particular problem was not included in \cite{Shoaei2014}.}
 \begin{tabular}{|c||c|c|c|c|} 
 \hline
 Model & PDRC & IISC\cite{Shoaei2014} & SC-BDD \\
 \hline\hline
 CMT(1,5) & $0.09$ & $0.13$ & $0.007$ \\ 
 \hline
 CMT(3,3) & $1.3$ & $0.43$ & $1.12$ \\ 
 \hline
 CMT(5,5) & $8.3$ & $0.73$ & $\times$ \\ 
 \hline
 CMT(7,7) & $30.02$ & $0.98$ & $\times$ \\ 
 \hline
 EDP(5,10) & $0.03$ & $0.98$ & $0.031$ \\
 \hline
 EDP(10,10) & $0.15$ & -- & $0.10$ \\
 \hline
 EDP(5,50) & $0.03$ & $0.12$ & $0.26$ \\
 \hline
 EDP(5,200) & $0.06$ & $0.12$ & $\times$ \\
 \hline
 EDP(5,10e3) & $0.19$ & $0.12$ & $\times$ \\
 \hline
 PME & $0.72$ & $2.3$ & $8.1$ \\ 
 \hline
\end{tabular}

\label{tab:results}
\end{table}


\section{Discussion} \label{sec:comparison}

In this section, we relate briefly how BDD-SC~\cite{BDD1,BDD2} and IISC~\cite{Shoaei2014} work, in order to compare and contrast to PDRC.

\subsection{BDD-SC}


BDD-SC works by modelling an FSM as a \emph{binary decision diagram} (BDD). The algorithm generates a BDD, representing the safe states, by searching backwards from the forbidden states. However, the size of this BDD grows with the domain of the integer variables. The reason is that the size of the BDD is quite sensitive to the number of binary variables, but also the ordering of the variables in the BDD. Even when more recent techniques on partitioning of the problem are used~\cite{BDD2}, the size of the BDD blows up, and we see in Table 1 that BDD-SC very quickly goes from good performance to time-out.

\subsection{IISC}

It is natural to compare PDRC to IISC~\cite{Shoaei2014}, since the latter is also inspired by PDR (albeit under the name IC3). In theory, PDRC has some advantages.


The first advantage is one of representation. IISC is built on the EFSM's separation between locations and variables, as described in \ref{sec:des-models}. PDRC, on the other hand, handles the more general STS representation. Specifically, IISC explicitly unrolls the entire sub-state-space spanned by the locations. This sub-space can itself suffer a space explosion when synchronising a large number of automata. 

To once again revisit our example (\Cref{fig:example}): IISC would unroll the graph, starting in $l_0$, into an \emph{abstract reachability tree}. Each node in such a tree can cover any combination of variable values, but only one location. Thus, IISC effectively does a forwards search for bad locations, and the full power of PDR (IC3) is only brought to bear on the assignment of variables along a particular error trace. Thus, a bad representation choice w.r.t. which parts of the system are encoded as locations versus as variables can hurt IISC, while PDRC is not so vulnerable. 

PDRC, in contrast, leverages PDR's combination of forwards and backwards search: exploring the state space backwards from the bad states in order to construct an inductive invariant which holds in the initial states. One disadvantage of the backwards search is that PDRC might add redundant safeguards. For example, the safeguard on the transition from $l_1$ to $1_3$ in \Cref{fig:example_controlled} is technically redundant, as there is no way to reach $l_2$ with the restricted variable values from the initial states. As shown in \cite{Shoaei2014}, IISC does not add this particular guard. However, since both methods are proven to yield minimally-restrictive supervisors, any extra guards added by PDRC are guaranteed not to affect the behaviour of the final system.

The gain, on the other hand, 
is that one does not need to unroll the whole path from the initial state to the forbidden state in order to supervise it. Consider: each such error path must have a ``point of no return''---the last controllable transition. When synthesising for safety, this transition must never be left enabled (our proof of \Cref{th:min_restrictive} hinges upon this). In order to find this point, PDRC traverses only the path between the point of no return and the forbidden state, whereas IISC traverses the whole path. In a sense, PDRC does not care about how one might end up close to forbidden state, but only where to put up the fence.



In practice, our results have IISC outperforming PDRC on both PDE and CMT. We believe the main reason is that unlike IISC which uses IC3 extended to SMT~\cite{Cimatti2012}, our implementation of PDRC works in SAT. This means that while both algorithms are theoretically equipped to abstract away large swathes of the state space, IISC does it much easier on integer variables than PDRC, which needs to e.g. represent each possible value of a variable as a separate gate. 

The one point where PDRC succeeds also in practice is on the PME problem. Here, most of the system's complexity comes from the number of different locations across the synchronised automata, rather than from large variable domains. In order to further explore this difference in problem type, we would have liked to evaluate PDRC and IISC on more problems with more synchronised automata, such as EDP(10,10). Sadly, this was impossible since the IISC implementation is no longer maintained.

\section{Conclusions and Future Work}

We have presented PDRC, an algorithm for controller synthesis of discrete event systems with uncontrollable transitions, based on property-driven reachability. The algorithm is proven to terminate on all solvable problem instances, and its synthesised controllers are proven to be safe and minimally restrictive. We have also implemented a prototype in the SAT-based model checker Tip. Our experiments show that even this SAT-based implementation outperforms a comparable BDD-based approach, 
but not the more recent IISC. However, since the implementation of IISC we compare against uses an SMT solver, not to mention that it is not maintained anymore, we must declare the algorithm-level comparison inconclusive.

The clearest direction for future research would be to implement PDRC using an SMT solver, to see if this indeed does realise further potential of the algorithm like we believe. Both \cite{Cimatti2012} and \cite{GPDR} provide good insights for this task. However, another interesting direction is to use both PDRC and IISC as a starting point to tackling the larger problem: safe \emph{and nonblocking} controller synthesis. Expanding the problem domain like this cannot be done by a trivial change to PDRC, but hopefully the insights from this work can contribute to a new algorithm. Another technique to draw from is that of IICTL~\cite{Hassan2012}. As discussed in \Cref{sec:desiderata}, by restricting our problem to only safety, we remove ourselves from real-world applications. For this reason, we do not present PDRC as a contender for any sort of throne, but as a stepping stone towards the real goal: formal, symbolic synthesis and verification of discrete supervisory control.



\newpage

\printbibliography




\end{document}